\documentclass[12pt,twoside,openany]{paper}
\usepackage{enumerate,amsmath,graphics,amssymb,graphicx,amscd,xypic,amscd,amsbsy,multirow,float,booktabs,verbatim,natbib}
\newcommand{\QED}{\hspace*{\fill}\rule{2.5mm}{2.5mm}}
\newtheorem{theorem}{Theorem}[section]
\newtheorem{definition}{Definition}[section]

\newenvironment{proof}{\noindent{\bf Proof\ }}{\QED\\}
\newcommand{\R}{\mathbb{R}}

\newcommand{\N}{\mathbb{N}}

\newtheorem{lemma}{Lemma}[section]

\begin{document}
\begin{center}
\vspace{0.5cm} {\large \bf ``Approximating quantiles in very large datasets"}\\
\vspace{1cm} Reza Hosseini, Simon Fraser University\\
Statistics and actuarial sciences, 8888 University Road,\\
Burnaby, BC, Canada, V65 1S6\\
 reza1317@gmail.com
\end{center}

\begin{abstract}
Very large datasets are often encountered in climatology, either
from a multiplicity of observations over time and space or outputs
from deterministic models (sometimes in petabytes= 1 million
gigabytes). Loading a large data vector and sorting it, is
impossible sometimes due to memory limitations or computing power.
We show that a proposed algorithm to approximating the median, ``the
median of the median'' performs poorly. Instead we develop an
algorithm to approximate quantiles of very large datasets which
works by partitioning the data or use existing partitions (possibly
of non-equal size). We show the deterministic precision of this
algorithm and how it can be adjusted to get customized precisions.
\end{abstract}


\section{Introduction}

This paper develops an algorithm for approximating the quantiles in
petascale (petabyte= one million gigabytes) datasets and uses the
``probability loss function'' to assess the quality of the
approximation. The need for such an approximation does not arise for
the sample average, another common data summary. That is because if
we break down the data to equal partitions and calculate the mean
for every partition, the mean of the obtained means is equal to the
total mean. It is also easy to recover the total mean from the means
of unequal partitions if their length is known.

However computer memories, several gigabytes (GBs) in size, cannot
handle large datasets that can be petabytes (PBs) in size. For
example, a laptop with 2 GBs of memory, using the well--known R
package, could find the median of a data file of about 150 megabytes
(MBs) in size. However, it crashed for files larger than this. Since
large datasets are commonly assembled in blocks, say by day or by
district, that need not be a serious limitation except insofar as
the quantiles computed in that way cannot be used to find the
overall quantile. Nor would it help to sub-sample these blocks,
unless these (possibly dependent) sub-samples could be combined into
a grand sub-sample whose quantile could be computed. That will not
usually be possible in practice. The algorithm proposed here is a
``worst-case'' algorithm in the sense that no matter how the data
are arranged, we will reach the desired precision. This is of course
not true if we sample from the data because there is a (perhaps
small) probability that the approximation could be poor.

We also address the following question:
\begin{quote}
{\it {\bf Question}: If we partition the data--file into a number of
sub-files and compute the medians of these, is the median of the
medians a good approximation to the median of the data--file?}
\end{quote}

We first show that the median of the medians does not approximate
the exact median well in general, even after imposing conditions on
the number of partitions or their length. However for our proposed
algorithm, we show how the partitioning idea can be employed
differently to get good approximations. ``Coarsening'' is introduced
to summarize data vector with the purpose of inferring about the
quantiles of the original vector using the summaries. Then the
``d-coarsening'' quantile algorithm which works by partitioning the
data (or use previously defined partitions) to possibly non-equal
partitions, summarizing them using coarsening and inferring about
the quantiles of the original data vector using the summaries. Then
we show the deterministic accuracy of the algorithm in Theorem
\ref{theorem-coarsening}. The accuracy is measured in terms of the
probability loss function of the original data vector. This is an
extension of the work in \cite{alsabti-1997} to non-equal size
partition case. Theorem \ref{theorem-coarsening} still requires the
partition sizes to be divisible by $d$ the coarsening factor. In
order to extend the results further to the case where the partitions
are not divisible by $d$, we investigate how quantiles of a data
vector with missing data or contaminated data relate to the
quantiles of the original data in Lemma
\ref{lemma-missed-data-quantile} and Lemma
\ref{lemma-polluted-data-quantile}. Also in Lemma
\ref{lemma-quantile-coarse}, we show if the quantiles of a coarsened
vector are used in place of the quantiles of the original data
vector how much accuracy will be lost. Finally we investigate the
performance of the algorithm using both simulations and real climate
datasets.

We define the loss of estimating/approximating a quantile $q$ by
$\hat{q}$ to be the probability that the random variable falls in
between the two values. A limited version of this concept only for
data vectors can be found in computer science literature, where
$\epsilon$-approximations are used to approximate quantiles of large
datasets. (See for example \cite{Manku98approximatemedians}.)
However, this concept has not been introduced as a measure of loss
and the definition is limited to data vectors rather than arbitrary
distributions.

The traditional definition of quantiles for a random variable $X$
with distribution function $F$,
\[lq_X(p)=\inf \{x|F(x) \geq p\},\]
appears in classic works as \cite{parzen-1979}. We call this the
``left quantile function''. In some books (e.g. \cite{rychlik}) the
quantile is defined as
\[rq_X(p)=\sup \{x|F(x) \leq p\},\]
this is what we call the ``right quantile function''. Also in
robustness literature people talk about the upper and lower medians
which are a very specific case of these definitions. \cite{reza-phd}
considers both definitions, explore their relation and show that
considering both has several advantages.

\begin{lemma} (Quantile Properties Lemma) Suppose $X$ is a random
variable on the probability space $(\Omega,\Sigma,P)$ with
distribution function $F$:

\begin{enumerate}[a)]
\item $F(lq_F(p))\geq p$. \item $lq_F(p) \leq rq_F(p)$.
\item $p_1<p_2 \Rightarrow rq_F(p_1)\leq lq_F(p_2)$.
\item $rq_F(p)=\sup\{x|F(x)\leq p\}$.
\item $P(lq_F(p)<X<rq_F(p))=0$. i.e. $F$ is flat in the interval $(lq_F(p),rq_F(p))$. \item $P(X<
rq_F(p)) \leq p$.
\item If $lq_F(p)<rq_F(p)$ then $F(lq_F(p))=p$ and hence $P(X\geq rq_F(p))=1-p$. \item $lq_F(1)>-\infty,rq_F(0)<\infty$
and $P(rq_F(0) \leq X \leq lq_F(1))=1$.
\item $lq_F(p)$ and $rq_F(p)$ are non-decreasing functions of $p$.
\item If $P(X=x)>0$ then $lq_F(F(x))=x.$ \item $x<lq_F(p) \Rightarrow F(x)<p$ and
$x>rq_F(p) \Rightarrow F(x)>p.$
\end{enumerate}

\label{quantile-properties}
\end{lemma}

\section{Previous work}

Finding quantiles and using them to summarize data is of great
importance in many fields. One example is the climate studies where
we have very large datasets. For example the datasets created by
computer climate models are larger than PBs in size. In NCAR
(National Center for Atmospheric sciences at Boulder, Colorado), the
climate data (outputs of compute models) are saved on several disks.
To access different parts of these data a robot needs to change
disks form a very large storage space. Another case where we
confront  large datasets is in dealing with data streams which arise
in many different applications such as finance and high--speed
networking. For many applications, approximate answers suffice. In
computer science, quantiles are important to both data base
implementers and data base users. They can also be used by business
intelligence applications to drive summary information from huge
datasets.

As pointed out  in \cite{Manku98approximatemedians}, a good quantile
approximation algorithm should
\begin{enumerate}
\item not require prior knowledge of the arrival or value distribution of its
inputs.
\item provide explicit and tunable approximation guarantees.
\item compute results in a single pass.
\item produce multiple quantiles at no extra cost.
\item use as little memory as possible.
\item be simple to code and understand.
\end{enumerate}
Finding quantiles of data vectors and sorting them are parallel
problems since once we sort a vector finding any given quantile can
be done instantly. A good account of early work in sorting
algorithms can be found in \cite{knuth-1973}. Also
\cite{munro-paterson-1980} showed for $P$-pass algorithms
(algorithms that scan the data $P$ times) $\Theta(N/P)$ storage
locations are necessary and sufficient, where $N$ is the length of
the dataset. (See Appendix C for the definitions of complexity
functions such as $\Theta$.) It is well-known that the worst-case
complexity of sorting is $n\log_2 n +O(1)$ as shown in
\cite{Manku99randomsampling}. In \cite{Paterson97progressin},
Paterson discusses the progress made in the so-called ``selection''
problem. He lets $V_k(n)$ be the worst-case minimum number of
pairwise comparisons required to find the $k$-th largest out of $n$
``distinct elements''. In particular $M(n)=V_k(n)$ for $k=\lceil n/2
\rceil$. In \cite{blum-1973}, it is shown that the lower bound for
$V_k(n)$ is $n+\min\{k-1,n-k\}-1$, an achieved upper bound by Blum
is $5.43n$. Better upper bounds have been achieved through the
years. The best upper bound so far is $2.9423N$ and the lower bound
is $(2+\alpha)N$ where $\alpha$ is of order $2^{-40}$.

 \cite{yao-1974} shows that finding approximate median needs
$\Omega(N)$ comparisons in deterministic algorithms. Using sampling
this can be reduced to $O(\frac{1}{\epsilon^2}\log(\delta^{-1}))$
independent of $N$, where $\epsilon$ is the accuracy of the
approximation in terms of the ``probability loss'' in our notation.
 \cite{munro-paterson-1980} show that $O(N^{1/p})$ is necessary
and sufficient to find an exact $\phi$-quantile in $p$ passes.

Often an exact quantile is not needed. A related problem is finding
space-efficient one-pass algorithms to find approximate quantiles. A
summary of the work done in this subject and a new method is given
in \cite{Agrawal95aone-pass}. Two approximate quantile algorithms
using only a constant amount of memory were given in
\cite{jain-1985} and \cite{Agrawal95aone-pass}. No guarantee for the
error was given.  \cite{alsabti-1997} provide an algorithm and
guaranteed error in one pass. This algorithm works by partitioning
the data into subsets, summarizing each partition and then finding
the final quantiles using the summarized partitions. The algorithm
in this chapter is an extension of this algorithm to the case of
partitions of unequal length.

\section{The median of the medians}

A proposed algorithm to approximate the median of a very large data
vector partitions the data into subsets of equal length, computes
the median for each partition and then computes the median of the
medians. For example, suppose $n=lm$ and break the data to $m$
vectors of size $l$. One might conjecture that by picking
 $l$ or $m$ sufficiently large the median of the medians would ensure close proximity to the exact median.
 We show by an example that taking $l$ and $m$ very
large will not help to get close to the exact median. Let $l=2b+1$
and $m=2a+1$.

\begin{example}
{\tiny
\begin{table}[H]
  \centering  \footnotesize
  \begin{tabular}{lccc}
\toprule[1pt]
   \mbox{partition number}& Partition & \mbox{Median of the partition}\\
  \midrule[1pt]
1 & $(1,2,\cdots,b,b+1,10^b,\cdots,10^b)$ & $b+1$ \\
2 & $(1,2,\cdots,b,b+1,10^b,\cdots,10^b)$ & $b+1$ \\
. & . & .\\
.& . & .\\
.& . & .\\
a & $(1,2,\cdots,b,b+1,10^b,\cdots,10^b)$ & $b+1$ \\
a+1 & $(1,2,\cdots,b,b+1,10^b,\cdots,10^b)$ & $10^b$\\
a+2 & $(10^b,10^b,\cdots,10^b)$  & $10^b$ \\
.& . & .\\
.& . & .\\
.& . & .\\
2a+1 & $(10^b,10^b,\cdots,10^b)$ & $10^b$ \\
 \bottomrule[1pt]
  \end{tabular}
  \caption{The table of data}
\label{table:first-number-median}
 \end{table}}
\end{example}

Table \ref{table:first-number-median} shows the dataset partitioned
into $m=2a+1$ vectors of equal length. Every vector is of length
$l=2b+1$. The first $a+1$ vectors are identical and $10^b$ is
repeated $b$ times in them. The last $a$ vectors are also identical
with all components equal to $10^b$. The median of the medians turns
out to be $b+1$. However, the median of the dataset is $10^b$.  We
show that $b+1$ is in fact ``almost" the first quantile. This is
because $(b+1)$ is smaller than all $10^b$'s. There are
$(a+1)b+a(2b+1)$ data points equal to $10^b$. Hence $b+1$ is smaller
than this fraction of the data points:

\[\frac{(a+1)b+a(2b+1)}{(2a+1)(2b+1)}=\frac{2a+2}{2a+1}\frac{b}{4b+2}+\frac{a}{2a+1}\approx 1\times \frac{1}{4}+\frac{1}{2} \approx \frac{3}{4}.\]
With a similar argument, we can show that $b+1$ is greater than
almost a quarter of the data points (the ones equal to
$1,2,\cdots,b$). Hence $b+1$ is ``almost" the first quantile.

One can prove a rigorous version of the the following statement.\\

{\it The median of the medians is ``almost" between the first and
the
third quartile.}\\

We only give a heuristic argument for simplicity. To that end, let
$n=lm$ and $m=2a+1$ and $l=2b+1$. Let $M$ be the exact median and
$M'$ be the median of the medians. Order the obtained medians of
each partition and denote them by $M_1,\cdots,M_m$. By definition
$M'\geq M_j,\; j \leq a$ and $M'\leq M_j,\; j \geq a+1.$ Each
$M_j,\;\;j \leq a$ is less than or equal to $b$ data points in its
partition. Hence, we conclude that $M'$ is less than or equal to
$ab$ data points. Similarly $M'$ is greater than or equal to $ab$
data points (which are disjoint for the data points used before).
But $\frac{ab}{n}=\frac{ab}{(2a+1)(2b+1)}\approx \frac{1}{4}$.
Hence, $M'$ is greater than or equal to 1/4 data points and less
than or equal to 1/4 data points.

\section{Preliminary results}

Suppose $y' \in \{y_1,\cdots,y_n\}$, for future reference, we define
some additional notations for data vectors.

\begin{definition} The minimal index of $y'$, $m(y')$ and the
maximal  index of $y'$, $M(y')$ are defined as below:
\[m(y')=\min\{i|y_i=y'\},\;M(y')=\max\{i|y_i=y'\}.\]
\end{definition}

It is easy to see that in $y=sort(x)=(y_1,\cdots,y_n)$ all the
coordinates between $m(y')$ and $M(y')$ are equal to $y'$. Also note
that if $y'=z_i$ then $M(y')-m(y')+1=m_i$ is the multiplicity of
$z_i$. We use the notation $m_x$ and $M_x$ whenever we want to
emphasize that they depend on the data vector $x$.

\begin{lemma} Suppose $x=(x_1,\cdots,x_n)$, $y=sort(x)$ and $z$ a
non--decreasing vector of all distinct elements of $x$. Then\\
a) $m(z_{i+1})=M(z_{i})+1,\;\;i=0,\cdots,r-1$.\\
b) Suppose $\phi$ is a bijective increasing transformation over
$\R$, \[m_\phi(x)(\phi(z_i))=m_{x}(z_i),\] and
\[M_{\phi(x)}(\phi(z_i))=M_{x}(z_i),\] for $\;i=1,\cdots,r.$
\label{m-M-properties}
\end{lemma}

\begin{proof}
a) is straightforward.\\
\noindent b) Note that
\[m_x(y')=\min\{i|y_i=y'\}=\min\{i|\phi(y_i)=\phi(y')\}=m_{\phi(x)}(\phi(y')).\]
A similar argument works for $M_x$.
\end{proof}

We also define the position and standardized position of an element
of a data vector.

\begin{definition}
Let $x=(x_1,\cdots,x_n)$ be a vector and
$y=sort(x)=(y_1,\cdots,y\_n)$. Then for $y' \in \{y_1,\cdots,y_n\}$,
we define
\[pos_x(y')=\{m_x(y'),m_x(y')+1,\cdots,M_x(y')\},\]
 where $pos$ stands for position. Then we define the standardized position of $y'$ to be

 \[spos_x(y')=(\frac{m_x(y')-1}{n},\frac{M_x(y')}{n}).\]
 \end{definition}
 In the following lemma we show that for every $p \in spos(y')$ (and only $p \in spos(y')$), we have
 $rq(p)=lq(p)=y'$. For example if $1/2 \in spos(y')$ then $y'$ is the
(left and right) median.

\begin{lemma} Suppose $x=(x_1,\cdots,x_n)$,
$y=sort(x)=(y_1,\cdots,y_n)$ and $y' \in \{y_1,\cdots,y_n\}$. Then

 \[p \in spos_x(y') \Leftrightarrow
lq_x(p)=rq_x(p)=y'.\]
\end{lemma}
\begin{proof}
Let $z=(z_1,\cdots,z_r)$ be the reduced vector with multiplicities
$m_1,\cdots,m_r.$ Then $y'=m_i$ for some $i=1,\cdots,r.$
\begin{enumerate}

\item[case I:] If $i=2,\cdots,r,$ then
\[m(y')=m_1+\cdots+m_{i-1}+1,\]
and
\[M(y')=m_1+\cdots+m_i.\]

\item[case II:] If $i=1$, then $m(y')=1$ and $M(y')=m_1$.

\end{enumerate}
In any of the above cases for $p \in
(\frac{m(y')-1}{n},\frac{M(y')}{n})$ and only $p \in
(\frac{m(y')-1}{n},\frac{M(y')}{n})$

\[rq_x(p)=lq_x(p)=z_i,\]
by definition.
\end{proof}

Now we prove a lemma. It is easy to see that if $u \in pos(y')$ then

 \[(\frac{u-1}{n},\frac{u}{n}) \subset spos(y').\]
 We conclude that
 \[\cup_{u \in pos(y')}(\frac{u-1}{n},\frac{u}{n}) \subset spos(y').\]
In fact $spos(y')$ can possibly have a few points on the edge of the
intervals not in $\cup_{u \in pos(y')}(\frac{u-1}{n},\frac{u}{n})$.

\begin{lemma}
Suppose $x$ is a data vector of length $n$ and $y'$ is an element of
this vector. Also assume

\[y'\geq x_i,\;\; i \in I, \;\;\; y' \leq x_j,\;\; j \in J,\]

\[I \cap J=\phi,\;\;\;I,J \subset \{1,2,\cdots,n\}.\]
Then there exist a $p$ in $(\frac{|I|-1}{n},1-\frac{|J|}{n})$ that
belongs to $spos(y')$. In other words $lq(p)=rq(p)=y'$.
\label{sp-lemma}
 \end{lemma}

\begin{proof}
From the assumption, we conclude that $pos(y')$ includes a number
between $|I|$ and $n-|J|$. Let us call it $u_0$. Hence
$(\frac{u_0-1}{n},\frac{u_0}{n}) \subset spos(y')$. Since $|I|\leq
u_0 \leq n-|J|$, we conclude that $spos(y')$ intersects with

\[\cup_{|I|\leq u \leq n-|J|} (\frac{u-1}{n},\frac{u}{n})\subset (\frac{|I|-1}{n},1-\frac{|J|}{n}).\]
\end{proof}

\section{A loss function to assess approximations of quantiles}

Our purpose is to find good approximations to the median and other
quantiles. We need a method to asses such approximations.  We
contend that such a method should not depend on the scale of the
data. In other words it should be invariant under monotonic
transformations. We define a function $\delta$ that measures a
natural ``degree of separation" between data points of
 a data vector $x$. For the sake of illustration, consider the example
$sort(x)=(1,2,3,3,4,4,4,5,6,6,7)$. Now suppose, we want to define
the degree of separation of 3,4 and 7 in this example. Since 4 comes
right after 3, we consider their degree of separation to be zero.
There are 3 elements between 4 and 7 so it is appealing to measure
their degree of separation as 3 but since the degree of separation
should be relative, we cab also divide by $n=11$, the length of the
vector, and get: $\delta(4,7)=3/11.$ We can generalize this idea to
get a definition for all pairs in $\R$. With the same example,
suppose we want to compute the degree of separation between 2.5 and
4.5 that are not members of the data vector. Then since there are 5
elements of the data vector between these two values, we define
their degree of separation as $5/11$. More formally, we give the
following definition.

\begin{definition} Suppose $x=(x_1,\cdots,x_n)$, a data vector and $z<z'$ let $\Delta_x(z,z')=\{i|z<x_i<z',i=1,\cdots,n\}$.
Then we define

\[\delta_x(z,z')=\frac{|\Delta_x(z,z')|}{n},\]
and $\delta_x(z,z)=0,$ where $|\Delta_x(z,z')|$ is the cardinality
of $\Delta_x(z,z')$. We call $\delta_x$ the ``degree of separation''
(DOS) or the ``probability loss function'' associated with $x$.
\end{definition}

We then have the following lemma about the properties of $\delta.$


\begin{lemma} The degree of separation $\delta_x$ has the following properties:

\begin{enumerate}[a)]
\item $\delta_x\geq 0.$

\item $y<y'<y'' \Rightarrow \delta_x(y,y'') \geq \delta_x(y,y').$

\item $\delta_{\phi(x)}(\phi(z),\phi(z'))=\delta_x(z,z')$ if $\phi$
is a strictly monotonic transformation.

\item $y=sort(x)$ and $y_i<y_j \Rightarrow \delta_{x}(y_i,y_j)\leq (j-i-1)/n.$
\end{enumerate}
\end{lemma}

\begin{proof}

Both a) and b) are straightforward.  To show (c), suppose $z<z'$ and
$\phi$ is strictly decreasing. (The strictly increasing case is
similar.) Then $\phi(z')<\phi(z)$ and hence
\[\Delta_{\phi(x)}(\phi(z),\phi(z'))=\{i| \phi(z')<\phi(x_i)<\phi(z)\}=\{i|z<x_i<z'\}=\Delta_x(z,z').\]
Finally d) is true because $|\Delta_x(y_i,y_j)|=|\{l|y_i<x_l<y_j,
l=1,\cdots,n\}| \leq j-i-1$.

\end{proof}

\noindent {\bf Remark.} The definition and results above can be
applied to random vectors $S=(X_1,\cdots,X_n)$ as well. In that
 $\delta_S(z,z')$ is  random.

\section*{Loss function for distributions}

We define a degree of separation for distributions which corresponds
to the notion of ``degree of separation'' defined for data vectors
to measure separation between data points.

\begin{definition} Suppose $X$ has a distribution function $ F$. Let
\[\delta_F(z',z)=\delta_F(z,z')=\lim_{u \rightarrow z^{-}}F(u)-F(z')=P(z'<X<z),\;\;z>z',\]
and $\delta_F(z,z)=0,\;z \in \R.$ We also denote this by $\delta_X$
whenever a random variable $X$ with distribution $F$ is specified.
We call $\delta_X$ the ``degree of separation'' or the ``probability
loss function'' associated with $X$.
\end{definition}

The following lemma is a straightforward consequence of the
definition.

\begin{lemma}
Suppose $x=(x_1,\cdots,x_n)$ is a data vector with the empirical
distribution $F_n$. Then
\[\delta_{F_n}(z,z')=\delta_x(z,z'),\;z,z' \in \R.\]
\end{lemma}
This lemma implies that to prove a result about the degree of
separation of data vectors, it suffices to show the result for the
degree of separation of random variables.

\begin{theorem}
Let $X,Y$ be random variables and $F_X,F_Y$, their corresponding
distribution functions.\\
a) Assume $Y=\phi(X),$ for a strictly increasing or decreasing
function $\phi:\R \rightarrow \R$. Then
$\delta_{F_X}(z,z')=\delta_{F_Y}(\phi(z),\phi(z')),\;z<z' \in \R.$\\
b) $\delta_{F_X}(z,z') \leq \delta_{F_X}(z,z''),\;z\leq \ z' \leq z''.$\\
c) $\delta_{F_X}(z_1,z_3) \leq \delta_{F_X}(z_1,z_2) +
\delta_{F_X}(z_2,z_3)+P(X=z_2)$.\\
d) Suppose, $p \in [0,1]$. Then $\delta_{F_X}(lq_{F_X}(p),rq_{F_X}(p))=0.$\\
e) Suppose, $p_1<p_2 \in [0,1].$ Then
$\delta_{F_X}(lq_{F_X}(p_1),rq_{F_X}(p_2))\leq p_2-p_1.$
\label{delta-prop}
\end{theorem}

\noindent {\bf Remark.} We may restate Part (c), for data vectors:
Suppose $x$ has length $n$ and $z_2$ is of multiplicity $m$, (which
can be zero). Then the inequality in (c) is equivalent to
$\delta_x(z_1,z_3) \leq \delta_x(z_1,z_2) +
\delta_x(z_2,z_3)+m/n$.\\

\begin{proof}

a) Note that for a strictly increasing function $\phi$, we have
\[P(z<X<z')=P(\phi(z)<\phi(X)<\phi(z')).\]
Now suppose $\phi$ is strictly decreasing. Then $z<z' \Rightarrow
\phi(z')<\phi(z)$. Let $Y=\phi(X)$. Then
\[\delta_X(z,z')=P(z<X<z')=P(\phi(z')<\phi(X)<\phi(z))=\delta_{Y}(\phi(z),\phi(z')).\]

b) This is trivial.

c) Consider the case $z_1<z_2<z_3$. (The other cases are easier to
show.) Then

\[\delta_{F_X}(z_1,z_3)=P(z_1<X<z_3)=P(z_1<X<z_2)+P(X=z_2)+P(z_2<X<z_3)\]

\[=\delta_{F_X}(z_1,z_2)+\delta_{F_X}(z_2,z_3)+P(X=z_2).\]

d) This result is a straightforward consequence of Lemma
\ref{quantile-properties} b) and c).

e) This result follows from
\[\delta_{F_X}(lq(p_1),rq(p_2))=P(lq(p_1)<X<rq(p_2))\]
\[=P(X<rq(p_2))-P(X \leq lq(p_1)) \leq p_2-p_1.\]
The last inequality being a result of Lemma
\ref{quantile-properties} a) and d).
\end{proof}

\noindent {\bf Remark}: (e),(b) immediately imply
\[\delta_{F_X}(lq_{F_X}(p_1),lq_{F_X}(p_2))\leq p_2-p_1,\]
and \[\delta_{F_X}(rq_{F_X}(p_1),lq_{F_X}(p_2))\leq p_2-p_1.\]

\noindent {\bf Remark.} We call part c) of the above theorem the
pseudo-triangle inequality.

\section{Data coarsening and quantile approximation algorithm}

This section introduces an algorithm to approximate quantiles in
very large data vectors. As we demonstrated in the previous section
the median of medians algorithm is not necessarily a good
approximation to the exact median of a data vector even if we have a
large number of partitions and large length of the partitions. The
algorithm is based on the idea of ``data coarsening" which we will
discuss shortly. The proposed algorithm can give us approximations
to the exact quantile of known precisions in terms of degree of
separation. After stating the algorithm, we prove some theorems that
give us the precision of the algorithm. The results hold for
partitions of non-equal length.

\begin{definition}
Suppose a data vector $x$ of length $n=n_1 n_2$ is given, $n_1,n_2>1
\in \N$. Also let $sort(x)=y=(y_1,\cdots,y_{n})$. Then the
$n_2$--coarsening of $x$, $C_{n_2}(x)$ is defined to be
$(y_{n_2},y_{2n_2},\cdots,y_{(n_1-1)n_2})$. Note that $C_{n_2}(x)$
has length $n_1-1$. Let $p_i=i/n_1,i=1,2,\cdots,(n_1-1).$ Then
$C_{n_2}(x)=(lq_x(p_1),\cdots,lq_x(p_{n_1-1}))$.
\end{definition}

We can immediately generalize the coarsening operator. Suppose
\[sort(x)=(y_1,\cdots,y_n),\] and $n_2<n$ is given. Then by The
Quotient--Remainder Theorem from elementary number theory, there
exist $n_1 \in \N \cup \{0\}$ and $r<n_2$ such that $n=n_1n_2+r$.
Define $C_{n_2}(x)=(y_{n_2},\cdots,y_{n_2(n_1-1)})$. The expression
is  similar to before. However, there are $n_2+r$ elements after
$y_{n_2(n_1-1)}$ in the sorted vector $y$. In this sense this
coarsening is not fully symmetric. We show that if $n_2$ is small
compared to $n$ this lack of symmetry has a small effect on the
approximation of quantiles.

Suppose $x$ is a data vector of length $n=\sum_{i=1}^m l_i$. We
introduce the coarsening algorithm to find approximations to the
large data vectors. \vspace{1cm}

 \noindent {\bf $d$-Coarsening quantiles
algorithm:}

\begin{enumerate}

\item Partition $x$ into vectors of length $l_1,\cdots,l_m$. (Or use
pre--existing partitions, e.g. partitions of data saved in various
files on the hard disk of a computer.)

\[x^1=(x_1,\cdots,x_{l_1}),x^2=(x_{l_1+1},\cdots,x_{l_1+l_2}),\cdots,x^m=(x_{\sum_{j=1}^{m-1}l_{j}+1},\cdots,x_n)\]

\item Sort each $x^l,\;l=1,2,\cdots,m$ and let
$y^l=sort(x^l),\;l=1,\cdots,m$:

\[y^1=(y_{1}^1,\cdots,y_{l_1}^1),\cdots,y^m=(y_{1}^m,\cdots,y_{l_m}^m).\]

\item $d$--Coarsen every vector:

\[(y_{d}^1,\cdots,y_{(c_1-1)d}^1),\cdots,(y_{d}^m,\cdots,y_{(c_m-1)d}^m),\]
and for simplicity drop $d$ and use the notation $w_{i}^j=y_{id}^j$.

\[w^1=(w_{1}^1,\cdots,w_{(c_1-1)}^1),\cdots,w^m=(w_{1}^m,\cdots,w_{(c_m-1)}^m).\]

\item Stack all the above vectors into a single vector and call it $w$. Find
$rq_w(p)$ (or $lq_w(p)$) and call it $\mu$. Then $\mu$ is our
approximation to $rq_x(p)$ (or $lq_x(p)$).
\end{enumerate}

\begin{theorem}
Suppose $x$ is of length $n=\sum_{i=1}^m l_i,\;m\geq 2$ and
$l_i=c_id$. Let $C=\sum_{i=1}^m c_i$. Apply the coarsening algorithm
to $x$ and find $\mu$ to approximate $rq_x(p)$ (or $lq_x(p)$). Then
$\mu$ is a (left and right) quantile in the interval
\[[p-\epsilon,p+\epsilon],\]
where $\epsilon=\frac{m+1}{C-m}$. In other words
$\delta_x(\mu,rq_x(p))\leq \epsilon$ and $\delta_x(\mu,lq_x(p)) \leq
\epsilon.$ When $l_i=cd,\;i=1,\cdots,m$,
$\epsilon=\frac{m+1}{m-1}\frac{1}{c-1}\leq\frac{3}{c-1}$.
\label{theorem-coarsening}
\end{theorem}

We need an elementary lemma in the proof of this theorem.
\begin{lemma} (Two interval distance lemma)\\
Suppose two intervals $I=[a,b]$ and $J=[c,d]$ subsets of $\R$ are
given. Then
\[\sup\{|p-q|,p \in I, q \in J\}=\max\{|a-d|,|b-c|\}.\]
\label{lemma-interval-dist}
\end{lemma}
\begin{proof}
$\sup\{|p-q|,p \in I, q \in J\} \geq \max\{|a-d|,|b-c|\}$ is trivial
because $a,b \in I$ and $c,d \in J$.\\
To show the converse note that $|p-q|=p-q$ or $q-p,\;p \in I, q \in
J$. But
\[p-q \leq b-c,\]
and
\[q-p \leq d-a.\]
Hence
\[|p-q| \leq \max\{b-c,d-a\} \leq \max\{|b-c|,|a-d|\}.\]
This completes the proof.
\end{proof}

\begin{proof} of Theorem \ref{theorem-coarsening}.\\

Let $n'=\sum_{i=1}^m (c_i-1)=\sum_{i=1}^m c_i-m=C-m$ and
$M_C=\{(i,j)|i=1,2\cdots,m,j=1,\cdots,c_i-1\}$, the index set of
$w$. Also let $c=\max\{c_1,\cdots,c_m\}.$

Suppose, $\frac{h-1}{n'} \leq p < \frac{h}{n'},\;h=1,\cdots,n'.$
Then since $\mu=rq_w(p)$, there are disjoint subsets of $M_C$, $K$
and $K'$ such that $|K|=h$, $|K'|=n'-h$, $\mu \geq w_j^i,\;(i,j)\in
K$ and $\mu \leq w_j^i,\;(i,j)\in K'$. (This is because if we let
$v=sort(w),$  $rq_w(p)=v_h$ since $[n'p]=h-1$.)

$K,K'$ are not necessarily unique because of possible repetitions
among the $w_t^i$. Hence we impose another condition on $K$ and
$K'$. If $(i,t) \in K$ then $(i,u) \notin K',\;u<t$. It is always
possible to arrange for this condition. For suppose,  $(i,t) \in K$
and $(i,u) \in K',\;u<t.$ Then $\mu \geq w_i^t$ and $\mu \leq
w_u^i,$ hence $w_t^i \leq w_i^u$. But since $u<t$ we have $w_t^i
\leq w_i^u$ by the definition of $w^i$. We conclude that
$w_t^i=w_i^u$. Now we can simply exchange $(i,t)$ and $(i,u)$
between $K$ and $K'$. If we continue this procedure after finite
number of steps we will get $K$ and $K'$ with the desired property.

Now define

\begin{itemize}

\item \[K_1=\{(i,1)|(i,1)\in K\},\]
with $|K_1|=k_1$ and
\[I_1=\{(i,j)|j \leq d,(i,1) \in K\},\]
Then $|I_1|=k_1 d$. Also note that if $(i,j) \in I_1,$  $\mu \geq
w_1^i \geq y_j^i$.

\item Let \[K_2=\{(i,2)|,(i,2)\in K\},\]
with $|K_2|=k_2$ and

\[I_2=\{(i,j)|d < j \leq 2d,(i,2) \in K\}.\]
Then $|I_2|=k_2 d$. Also note that if $(i,j) \in I_2,$ $\mu \geq
w_2^i \geq y_j^i$.

\item Let \[K_t=\{(i,t)|(i,t)\in K\},\]
with $|K_t|=k_t$ and

\[I_t=\{(i,j)|(t-1)d < j \leq td,(i,t) \in K\}.\]
Then $|I_t|=k_t d$. Also note that if $(i,j) \in I_t,$ $\mu \geq
w_t^i \geq y_j^i$.

\item Let \[K_{c-1}=\{(i,(c-1))|(i,c-1)\in K\},\]
with $|K_{c-1}|=k_{c-1}$ and

\[I_{(c-1)}=\{(i,j)|(c-2)d < j \leq (c-1)d,(i,c-1) \in K\}.\]
Then $|I_{c-1}|=k_{c-1} d$.  Also note that if $(i,j) \in
I_{(c-1)},$  $\mu \geq w_{(c-1)}^i \geq y_j^i$.
\end{itemize}

Note that $K=\cup_{t=1}^{c-1} K_t,\;\;|K|=k_1,+\cdots+k_{c-1}$.
Since the $K_t$ are disjoint the $I_t$ are also disjoint. Let
$I=\cup_{t=1}^{c-1} I_t$ then $|I|=d(k_1+\cdots+k_{c-1})=d|K|$. Also
note that $(i,j) \in I \Rightarrow \mu \geq y_j^i$.

Similarly define,

\begin{itemize}

\item \[K'_1=\{(i,1)|(i,1)\in K'\},|K'_1|=k'_1,\]
and
\[I'_1=\{(i,j)|d < j \leq 2d,(i,1) \in K'\}.\]
Then $|I'_1|=k'_1 d$. Also note that if $(i,j) \in I'_1,$  $\mu \leq
w_1^i \leq y_j^i$.

\item Let \[K'_2=\{(i,2)|(i,2)\in K'\},|K'_2|=k'_2,\]
and
\[I'_2=\{(i,j)|2d < j \leq 3d,(i,2) \in K'\}.\]
Then $|I'_2|=k'_2 d$. Also note that if $(i,j) \in I'_2,$ $\mu \leq
w_2^i \leq y_j^i$.

\item Let \[K'_t=\{(i,t)|(i,t)\in K'\},|K'_t|=k't,\]
and
\[I'_t=\{(i,j)|td < j \leq (t+1)d,(i,t) \in K'\}.\]
Then $|I'_t|=k'_t d$. Also note that if $(i,j) \in I'_t$ then $\mu
\leq w_t^i \leq y_j^i$.

\item \[K'_{c-1}=\{(i,(c-1))|(i,c-1)\in K'\},|K'_{c-1}|=k'_{c-1},\]
and
\[I'_{c-1}=\{(i,j)| j > (c-1)d,(i,c-1) \in K'\}.\]
Then $|I'_{c-1}|=k'_{c-1} d$. Also note that if $(i,j) \in I'_{c-1}
\Rightarrow \mu \leq w_{(c-1)}^i \leq y_j^i$.
\end{itemize}
Then $|I|=|K|d$ and $|I'|=|K'|d$. We claim that $I \cap
I'=\emptyset$. To see this note that because of how the second
components in $I_t$ and $I_t'$ are defined, it is only possible that
$I_{t+1}=\{(i,j)|td < j \leq (t+1)d,(i,t+1) \in K\}$ and
$I_t'=\{(i,j)|td<j \leq (t+1)d,(i,t) \in K'\}$ intersect for some
$t=1,\cdots,c-2$. But if they intersect then there exist $i,t$ such
that $(i,t+1) \in K$ and $(i,t) \in K'$ which is against our
assumption regarding $K$ and $K'.$ Hence by Lemma \ref{sp-lemma},
$\mu$ is a quantile between

\[[\frac{|K|d}{n},\frac{n-|K'|d}{n}]=[\frac{hd}{\sum_{i=1}^m c_i d},\frac{n-(n'-h)d}{\sum_{i=1}^m c_i d}]
=[\frac{h}{C},\frac{m+h}{C}].\] But we know that
\[p \in [\frac{h-1}{C-m},\frac{h}{C-m}).\]
We are dealing with two interval in one of them $\mu$ is a quantile
and the other contains $p$.

We showed in Lemma \ref{lemma-interval-dist} if two intervals
$[a,b]$ and $[c,d]$ are given, the sup distance between two elements
of the two intervals is

\[\max \{|a-d|,|b-c|\}.\]
Applying this to the above two intervals we get,

\[\max \{|\frac{m+h}{C}-\frac{h-1}{C-m}|,
|\frac{h-1}{C-m}-\frac{h}{C}|\},\] which is equal to,
\[\max \{|\frac{mC-m^2-hm+C}{C(C-m)}|,
|\frac{C-hm}{C(C-m)}|\}.\] But $m^2+hm \leq m^2 + (C-m)m=mC$. Hence

\[|\frac{mC-m^2-hm+C}{C(C-m)}|=\frac{mC-m^2-hm+C}{C(C-m)} \leq \frac{mC+C}{C(C-m)}=\frac{m+1}{C-m}.\]
Also
\[|\frac{C-hm}{C(C-m)}| \leq \frac{C+mC}{C(C-m)}\leq \frac{m+1}{C-m}.\]
Hence the max is smaller than $\epsilon=\frac{m+1}{C-m}$ and we
conclude that $\mu$ is a quantile for $p'$ which is at most as far
as $\epsilon$ to $p$.

The case $l_i=cd$ is easily obtained by replacing $C=mc$ and noting
that $\frac{m+1}{m-1}\leq 3,\;\;m\geq 2.$
\end{proof}

In most applications, usually the data partitions are not divisible
by $d$. For example the data might be stored in files of different
length with common factors. Another situation involves a very large
file that is needed to be read in successive stages because of
memory limitations. Suppose that we need a precision $\epsilon$ (in
terms of degree of separation) and based on that we find an
appropriate $c$ and $m$. Note that $n$ might not be divisible by
$mc$.

First we prove two lemmas. These lemmas show what happens to the
quantiles if we throw away a small portion of the data vector or add
some more data to it. The first lemma is for a situation that we
have thrown away or ignored a small part of the data. The second
lemma is for a situation that a small part of the data are
contaminated or includes outliers. In both cases, we show how the
quantiles computed in the ``imperfect'' vectors correspond to the
quantiles of the original vector. In both case $x$ stands for the
imperfect vector and $w$ is the complete/clean data.

\begin{lemma} (Missing data quantile approximation lemma)\\
Suppose $x=(x_1,\cdots,x_n),$ $sort(x)=(y_1,\cdots,y_n)$ and
$y'=lq_x(p), p \in [0,1].$ Consider a vector $x^{\star}$ of length
$n^{\star}$ and let $w=stack(x,x^{\star})$. Then $y'=lq_w(p')$,
where $p' \in [p-\epsilon,p+\epsilon]$ and
$\epsilon=\frac{n^{\star}}{n+n^{\star}}.$

Similarly if $y'=rq_x(p)$ and $p \in [0,1],$ $y'=rq_w(p')$, where
$p' \in [p-\epsilon,p+\epsilon]$ and
$\epsilon=\frac{n^{\star}}{n+n^{\star}}$.
\label{lemma-missed-data-quantile}
\end{lemma}

\begin{proof}
We prove the result for $lq_x$ only and a similar argument works
for $rq_x$.\\
 Let $z=sort(w)$ then $lq_z=lq_w$. For $p=1$ the result is easy to see. Otherwise,  $\frac{i}{n}
\leq p < \frac{i+1}{n}$ for some $i=0,\cdots,n-1$. But then
$y'=lq_x(p)=y_i$. In the new vector $z$ since we have added
$n^{\star}$ elements $y'=z_j$ for some $j,$ $i \leq j <
i+n^{\star}$. Hence $y'=lq_z(\frac{j}{n+n^{\star}})$. From $np-1 < i
\leq np $, we conclude

\[\frac{np-1}{n+n^{\star}}<\frac{i}{n+n^{\star}} \leq \frac{j}{n+n^{\star}}<\frac{i+n^{\star}}{n+n^{\star}} \leq \frac{np+n^{\star}}{n+n^{\star}}.\]
Hence,
\[\frac{n^{\star}(1-p)-1}{n+n^{\star}}<\frac{j}{n+n^{\star}}-p<\frac{n^{\star}(1-p)}{n+n^{\star}} \Rightarrow\]
\[|\frac{j}{n+n^{\star}}-p|<\max\{|\frac{n^{\star}(1-p)-1}{n+n^{\star}}|,|\frac{n^{\star}(1-p)}{n+n^{\star}}|\}.\]
But $|\frac{n^{\star}(1-p)}{n+n^{\star}}| \leq
\frac{n^{\star}}{n+n^{\star}}$ and
$|\frac{n^{\star}(1-p)-1}{n+n^{\star}}| \leq
\max\{\frac{n^{\star}-1}{n+n^{\star}},\frac{1}{n+n^{\star}}\}$ since
$p$ ranges in $[0,1]$. We conclude that that
\[|\frac{j}{n+n^{\star}}-p|<\frac{n^{\star}}{n+n^{\star}}.\]

\end{proof}

\begin{lemma} (Contaminated data quantile approximation lemma)\\
Suppose $x=(x_1,\cdots,x_n)$, $sort(x)=(y_1,\cdots,y_n)$ and
$y'=lq_x(p), p \in [0,1].$ Consider the vector
$w=(x_1,x_2,\cdots,x_{n-n^{\star}})$ then $y'=lq_{w}(p')$, where $p'
\in [p-\epsilon,p+\epsilon]$ and
$\epsilon=\frac{n^{\star}}{n-n^{\star}}.$

Similarly if $y'=rq_x(p)$ and $p \in [0,1],$ $y'=rq_w(p')$, where
$p' \in [p-\epsilon,p+\epsilon]$ and
$\epsilon=\frac{n^{\star}}{n-n^{\star}}$.
\label{lemma-polluted-data-quantile}
\end{lemma}

\begin{proof}
We only show the case for $lq_x$ and a similar argument works for
$rq_x$.\\
Let $z=sort(w).$ Then $lq_z=lq_w.$ If $p=1$ the result is easy to
see. Otherwise, $\frac{i}{n} \leq p < \frac{i+1}{n}$ for some
$i=0,\cdots,n-1$. But then $y'=lq_x(p)=y_i$. In the new vector $z$
since we have removed $n^{\star}$ elements $y'=z_j$ for some $j,$
$i-n^{\star} \leq j \leq i$. Hence $y'=lq_z(\frac{j}{n-n^{\star}})$.
From $np-1 < i \leq np $, we conclude $np-1-n^{\star} < j \leq
np\Rightarrow np-n^{\star} \leq j \leq np.$ Hence

\[ \frac{-n^{\star}+n^{\star}p}{n-n^{\star}}\leq \frac{j}{n-n^{\star}}-p \leq \frac{n^{\star}p}{n-n^{\star}}\Rightarrow \]
\[|\frac{j}{n-n^{\star}}-p| \leq \frac{n^{\star}}{n-n^{\star}}.\]

\end{proof}

 In the case that the partitions are not divisible by $d$, we can use the same algorithm with generalized
coarsening. The error will increase obviously and the next two
lemmas say by how much.

\begin{lemma}
Suppose $x$ has length $n=lm+r$, $0 \leq r<l$ and $m=cd$. To find
$lq_x(p)$, apply the algorithm in the previous theorems to a
sub--vector of $x$ of length $lm$. Then the obtained quantile is a
quantile for a number in $[p-\epsilon,p+\epsilon]$, where
$\epsilon=\frac{m+1}{m-1}\frac{1}{c-1}+\frac{r}{lm+r}$.
\end{lemma}

\begin{proof}
The result is a straightforward consequence of the Theorem
\ref{theorem-coarsening} and the Lemma
\ref{lemma-missed-data-quantile}.
\end{proof}

\begin{lemma}
Suppose $x$ has length $n=\sum_{i=1}^m l_i$ and $l_i=c_id+r_i$,
$r_i<d$. Let $R=\sum_{i=1}^m r_i$. Then apply the algorithm above to
$x$ to find $lq_x(p)$, using the generalized coarsening. The
obtained quantile is a quantile for a number in
$[p-\epsilon,p+\epsilon]$ where
$\epsilon=\frac{m+1}{C-m}+\frac{R}{R+Cd}.$
\end{lemma}

\begin{proof}
Let $l_i'=c_id$. Consider $x'$ a sub-vector of $x$ consisting of
\[(y_1^1,\cdots,y_{l_1'}^1),(y_1^2,\cdots,y_{l_2'}^2),\cdots,(y_1^m,\cdots,y_{l_m'}^m).\]
Then $x'$ has length $\sum_{i=1}^m l_i'$. By Lemma
\ref{lemma-missed-data-quantile} $p$-th quantile found by the
algorithm is a quantile in
$[p-\epsilon_1,p+\epsilon_1],\;\epsilon_1=\frac{m+1}{C-m}$ for $x'$.
$x$ has $R=\sum_{i=1}^m r_i$ elements more than $x'$. Hence the
obtained quantile is a quantile for $x$ for a number in
$[p-\epsilon,p+\epsilon]$, $\epsilon=\epsilon_1+\frac{R}{R+Cd}.$
\end{proof}

\section{Applications and computations}

Suppose a data vector $x$ has length $n$. To find the quantiles of
this vector, we only need to sort $x$. Since then for any $p \in
(0,1)$, we can find the first $h$ such that $p\geq h/n$. Note that

\[sort(x)=(lq_x(1/n),lq_x(2/n),\cdots,lq_x(1))=(rq_x(0),rq_x(1/n),\cdots,rq_x(\frac{n-1}{n})).\]
We only focus on left quantiles here. Similar arguments hold for the
right quantile.

Obviously, the longer the vector $x$, the finer the resulting
quantiles are. Now imagine that we are given a very long data vector
which cannot even be loaded on the computer memory. Firstly, sorting
this data is a challenge and secondly, reporting the whole sorted
vector is not feasible. Assume that we are given the sorted data
vector so that we do not need to sort it. What would be an
appropriate summary to report as the quantiles? As we noted also the
sorted vector itself although appropriate, maybe of such length as
to make further computation and file transfer impossible. The
natural alternative would be to coarsen the data vector and report
the resulting coarsened vector. To be more precise, suppose,
$length(x)=n=n_1n_2$ and $y=sort(x)=(y_1,\cdots,y_n)$. Then we can
report
\[y'=C_{n_2}(y)=(y_{n_2},\cdots,y_{(n_1-1)n_2}).\]
This corresponds to

\[(lq_{y'}(1/n_2),\cdots,lq_{y'}(1)).\]
How much will be lost by this coarsening? Suppose, we require the
left quantile corresponding to $(h-1)/n<p \leq h/n,\;h=1,\cdots,n$.
Then $x$ would give us $y_{h}$. But since $(h-1)/n<p \leq h/n$

\[np< h \leq np+1.\]
Also suppose for some $h'=1,\cdots,n_1,$ \[(h'-1)/(n_1-1)<p \leq
(h')/(n_1-1)\Rightarrow (h'-1)<p(n_1-1) \leq h'\]
\[\Rightarrow (n_1-1)p \leq h'< p(n_1-1)+1. \] Then
\[(h-1)(n_1-1)/n <h'< h(n_1-1)/n+1,\]
and
\begin{equation}
(h-1)(n_1-1)n_2/n <h'n_2< h(n_1-1)n_2/n+n_2.
\label{eq-quantile-coarsen}
\end{equation}
Using the coarsened vector, we would report $y_{h'(n_2)}$ as the
approximated quantile for $p$. The degree of separation between this
element and the exact quantile using Equation
\ref{eq-quantile-coarsen} is less than or equal to

\[\max \{\frac{|h-(h-1)(n_1-1)n_2/n|}{n},\frac{|h(n_1-1)n_2/n+n_2-h|}{n}\}.\]
This equals
\[\max
\{|\frac{-hn_2-n_1n_2+n_2}{n^2}|,|\frac{-hn_2+nn_2}{n^2}|\}.\] But

\[|\frac{-hn_2-n_1n_2+n_2}{n^2}|=\frac{n_2(n_1+n-1)}{n^2}<\frac{n_2(n_1+n)}{n^2}=\frac{1}{n}+\frac{n_2}{n},\]
and
\[|\frac{-hn_2+nn_2}{n^2}|<\frac{n_2}{n}.\]
Hence the degree of separation is less than $1/n+1/n_1$. We have
proved the following lemma.

\begin{lemma}
Suppose $x$ is a data vector of the length $n=n_1n_2$ and
$y=sort(x)$, $y'=C_{n_2}(y)$. Then if we use the quantiles of $y'$
in place of $x$, the accuracy lost in terms of the probability loss
of $x$ ($\delta_x$) is less than $1/n+1/n_1$.
\label{lemma-quantile-coarse}
\end{lemma}

The algorithm proposes that instead of sorting the whole vector and
then coarsening it, coarsen partitions of the data. The accuracy of
the quantiles obtained in this way is given in the theorems of the
previous section. This allows us to load the data into the memory in
stages and avoid program failure due to the length of the data
vector. We are also interested in the performance of the method in
terms of speed, and do a simulation study using the ``R'' package (a
well--known software for statistical analysis) to assess this. In
order to see theoretical results regarding the complexity of the
special case of the algorithm for equal partitions see
\cite{alsabti-1997}. For the simulation study, we create a vector,
$x$, of length $n=10^7$. We apply the algorithm for
$m=1000,c=20,d=500$. We create this vector in a loop of length 1000.
During each iteration of the loop, we generate a random mean for a
normal distribution by first sampling from $N(0,100)$. Then we
sample 10,000 points from a normal distribution with this mean and
standard deviation 1. We compare two scenarios:

\begin{enumerate}
\item Start by a NULL vector $x$ and in each iteration add the full generated vector of length 10000 to $x$. After the
loop has completed its run, sort the data vector which now has
length $10^7$ by the command sort in R and use this to  find the
quantiles.

\item Start with a NULL vector $w$. During each iteration after
generating the random vector, $d$-coarsen the data by $d=500$.
(Hence $m=1000$, $c=20$.) In order to do that computing, first apply
the sort command to the data and then simply $d$-coarsen the
resulting sorted vector. During each iteration, add the coarsened
vector to $w$. After all the iterations, sort $w$ and use it to
approximate quantiles.
\end{enumerate}

\noindent {\bf Remark.} The first part corresponds to the
straightforward quantiles' calculation and the second corresponds to
our algorithm. Note that in the real examples instead of the loop,
we could have a list of 1000 data files and still this example
serves as a way of comparing the straightforward method and our
algorithm.

\noindent {\bf Remark.} Note that if we wanted to create an even
longer vector say of length $10^{10}$ then the first method would
not even complete because the computer would run out of memory in
saving the whole vector $x$.

\noindent {\bf Remark.} The final stage of the algorithm can use the
fact that $w$ is built of ordered vectors to make the algorithm even
faster. We will leave that a problem to be investigated in the
future.

We have repeated the same procedure for $n=2 \times
10^7,m=1000,d=500$ and $n=10^8,m=1000,d=500$. The results of the
simulation are given in Table \ref{table-quantile-approx-sim}, in
which ``DOS'' stands for the degree of separation between the exact
median and the approximated median. The ``DOS bound" bounds the
degree of separation obtained by the theorems in the previous
section. For $n=10^7,n=2 \times 10^7$ significant time accrue by
using the algorithm. For a vector of length $10^8$, R crashed when
we tried to sort the original vector and only the algorithm could
provide results. For all cases the exact and approximated quantiles
are close. In fact the dos is significantly smaller than the dos
bound. This is because this is a ``worst-case" bound. The exact and
approximated quantiles for $n=10^7$ are plotted in Figure
\ref{quantiles-sim-compare.ps}.

{\tiny
\begin{table}[H]
\centering  \footnotesize
 \begin{tabular}{lcccc}
\toprule[1pt] Length & $n=10^7$ & $n=2 \times 10^7$ & $n=10^8$ \\
Exact median value & 1.847120 &  1.857168 & NA \\

 Algorithm median value &  1.866882 & 1.846463 & 1.846027 \\

 DOS  & 0.00012 & $-6.475 \times 10^{-5}$  & NA \\

  DOS bound & 0.05268421 & 0.02566667 & 0.005030151  \\
  Time for exact median & 186 sec & 461 s  & NA \\

 Time for the algorithm & 6 sec & 18 s & 98 s \\
\bottomrule[1pt]
  \end{tabular}
  \caption{Comparing the exact method with the proposed algorithm in R run on a laptop
  with 512 MB memory and a processor 1500 MHZ, $m=1000,d=500$. ``DOS'' stands for degree of separation in the original vector.
  ``DOS bound'' is the
theoretical degree of separation obtained by Theorem
\ref{theorem-coarsening}.} \label{table-quantile-approx-sim}
 \end{table}}

\begin{figure}
\centering
\includegraphics[width=0.55\textwidth] {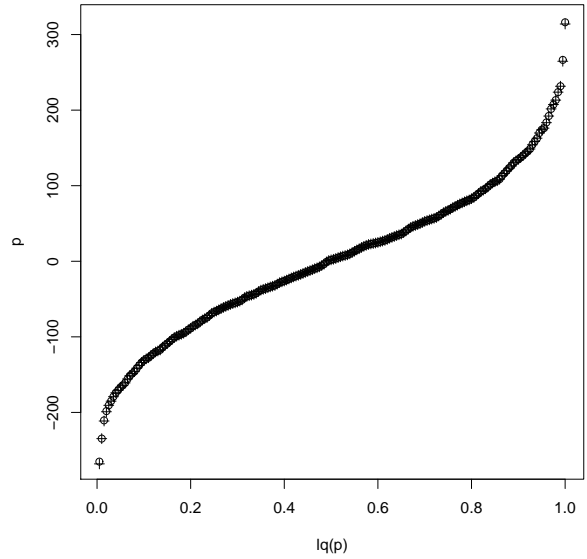}
 \caption{Comparing the approximated quantiles to the exact quantiles $N=10^7$. The circles are the exact quantiles and the
$+$ are the corresponding approximated quantiles.}
 \label{quantiles-sim-compare.ps}
\end{figure}

\begin{figure}
\centering
\includegraphics[width=0.55\textwidth] {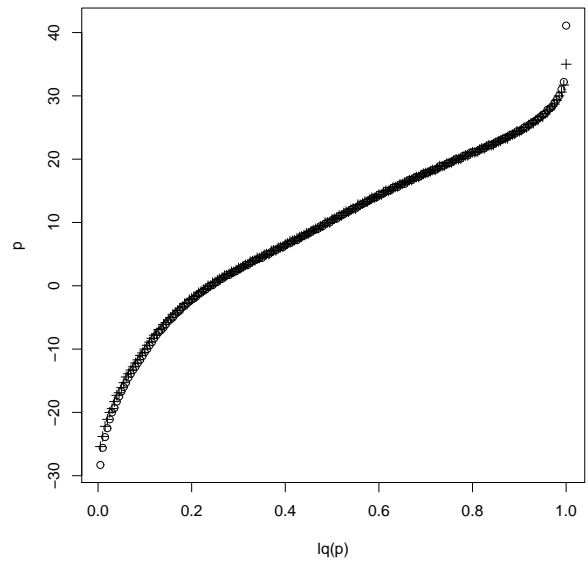}
 \caption{Comparing the approximated quantiles to the exact quantiles for $MT$ (daily maximum temperature) over 25 stations in Alberta 1940--2004.
The circles are the exact quantiles and the $+$ the approximated
quantiles.}
 \label{quantiles-compare-MT.ps}
\end{figure}

Next, we apply the algorithm on a real dataset. The dataset includes
the daily maximum temperature for 25 stations over Alberta during
the period 1940--2004. We focus on the 95th percentile. The results
are given in Table \ref{table-quantile-approx-MT}. The algorithm
finds the percentile more quickly but the time difference is not as
large as the simulation. This is because most of the time of the
algorithm and the exact computation is spent on reading the files
from the hard drive. The dos bound is about 0.01 (on the 0--1
probability scale). The true degree of separation is about 0.001.
The estimated quantiles and the exact quantiles are plotted in
Figure \ref{quantiles-compare-MT.ps}. Notice that the exact and
approximated values match except at the very beginning (very close
to zero) and end (when it is close to 1), where we see that the
circles (corresponding to exact quantiles) and the $+$s
(corresponding to the approximated quantiles) do not completely
match. This difference is at most 0.01 in terms of dos in any case.

{\tiny
\begin{table}[H]
\centering  \footnotesize
 \begin{tabular}{lc}
\toprule[1pt]

Exact 95th percentile &  27 C\\

 Algorithm 95th percentile &  26.7 C\\

 DOS  & 0.001278726 \\

  DOS bound & 0.01052189\\
  time for exact median & 8 min 6 sec\\

 time for the algorithm & 7 min 29 sec\\

\bottomrule[1pt]
  \end{tabular}
  \caption{Comparing the exact method with the proposed algorithm in R (run on a laptop
  with 512 MB memory and processor 1500 MHZ) to compute the quantiles of $MT$ (daily maximum temperature) over 25 stations
with data from 1940 to 2004.}
 \label{table-quantile-approx-MT}
 \end{table}}

\noindent{\bf Acknowledgements}: I would like to thank Jim Zidek and
Nhu Le for insightful comments and Jim Zidek for bringing up the
motivating question.

\bibliographystyle{plainnat}
\bibliography{mybibreza}

\end{document}